\documentclass[a4paper,12pt]{article}
\usepackage[english]{babel}
\usepackage[letterpaper,top=2cm,bottom=2cm,left=3cm,right=3cm,marginparwidth=1.75cm]{geometry}

\usepackage{amsmath,amssymb, mathtools,amsthm}
\usepackage{graphicx}
\usepackage{physics}
\usepackage[colorlinks=true, allcolors=blue]{hyperref}
\usepackage{booktabs}
\newcommand{\ir}{\mathtt{CER}}
\usepackage{algorithm}
\usepackage[noend]{algorithmic}
\usepackage{todonotes}
\usepackage{authblk}
\usepackage{xspace}

\newtheorem{theorem}{Theorem}

\title{A generic multi-Pauli compilation framework for limited connectivity}
\author{Adam Glos\thanks{adam.glos@algorithmiq.fi}\xspace}
\author{\"Ozlem Salehi}
\affil{Algorithmiq Ltd, Kanavakatu 3C 00160 Helsinki, Finland}
\date{}
\begin{document}
\maketitle

\begin{abstract}
Efficient and effective compilation of quantum circuits remains an important aspect of executing quantum programs. In this paper, we propose a generic compilation framework particularly suitable for limited connectivity, that extends many of the known techniques for Pauli network synthesis. Our compilation method, built on the introduced Clifford Executive Representation, stands out by considering the implementation of multiple Pauli operators at once, which are not necessarily commuting. The proposed technique also allows for effective Clifford circuit synthesis. We benchmark our methods against circuits resulting from the Variational Quantum Eigensolver algorithm and the results show that the proposed methods outperform the state-of-the-art.
\end{abstract}

\section{Introduction}

Quantum computing holds the promise to revolutionize many fields due to the intrinsic properties of quantum mechanics. Among the quantum computational models, gate-based quantum computers have gained notable interest. For gate-based quantum computers, the design procedure of a quantum algorithm involves the translation of a program written in a high-level quantum programming language into a series of 1 and 2-qubit gates forming the quantum circuit, which is known as the process of compilation. One of the challenges lies in how efficiently this process can be achieved, as minimizing the number of such gates increases the robustness of the quantum circuit, which is particularly important for the current Noisy Intermediate Scale Quantum (NISQ) \cite{preskill2018quantum} devices, but also decreases the execution time of the quantum circuits which is relevant for both NISQ and fault-tolerant eras.

One of the fields that quantum computing has the potential to transform is quantum simulation,  which is an idea that already arose in 1980s in the paper ``Simulating Physics with Computers'' by Feynman \cite{feynman2018simulating}. Simulating molecular systems is crucial for understanding their properties and their behavior in various chemical reactions. This is achieved by finding the ground state of an Electronic Structure Problem Hamiltonian, which becomes intractable for classical computational chemistry methods for large molecules. Obtaining more accurate and efficient calculations of electronic structures is essential for many fields including drug discovery \cite{cao2018potential, maniscalco2022quantum} and material science \cite{bauer2020quantum}, and quantum computing has the potential to transform those fields.

One of the most celebrated quantum algorithms that can approximate ground state energy is quantum phase estimation (QPE)~\cite{abrams1999quantum}. However, practical implementation of QPE is not possible with current noisy intermediate-scale quantum computers due to the substantial amount of gate and qubit requirements. One promising approach to harness the potential NISQ devices is the Variational Quantum Eigensolver (VQE) algorithm proposed by Peruzzo et al. \cite{peruzzo2014variational}, which relies on the variational principle to estimate the ground state energy. In VQE, a parametrized circuit, i.e. an ansatz is constructed and the expectation value is estimated by measuring the quantum state w.r.t. Hamiltonian. This quantity is then minimized by optimization of the parameters on a classical computer. The choice of the ansatz should take into account both the expressibility of the quantum state and the required resources, as the depth and number of gates should be tractable for the near-term devices. Popular choices include unitary coupled cluster with single and doubles (UCCSD) \cite{peruzzo2014variational,romero2018strategies}, hardware efficient ansatz \cite{kandala2017hardware} and iterative ansatz construction methods like Adaptive Derivative-Assembled PseudoTrotter (ADAPT) VQE \cite{grimsley2019adaptive} and its variants \cite{tang2021qubit, yordanov2021qubit, ramoa2024reducing, anastasiou2024tetris, nykanen2022mitigating}. Independent of the chosen method, the VQE ansatz can often be seen as a series of operators $e^{-i \theta P}$ where $P$ is a Pauli string. Once constructed, it should undergo an efficient compilation procedure obeying the constraints of the underlying architecture. The goal is to reduce two-qubit entangling gates like CNOT gates, as such gates have lower fidelities compared to single-qubit gates on the current NISQ hardware. 

Various methods have been proposed to improve the compilation procedure for general Pauli strings. Some of the works use gate cancellation between the subcircuits implementing consecutive Pauli strings \cite{li2022paulihedral, jin2024tetris} and insertion of swap gates to match limited connectivity. Commutativity between the Pauli strings has been employed for simultaneous diagonalization of the Pauli strings in  \cite{van2020circuit, cowtan2020generic}, yet those works consider only all-to-all connectivity and can be used only for commutative Pauli strings. In the thesis \cite{meijer2024advances}, it is suggested as a future work that the method presented in \cite{cowtan2020generic} can be modified to be architecture aware using Steiner trees. 

The methods above work by implementing a single or a group of Paulis simultaneously, but eventually they `reset' the state of the quantum system. Thus one can see this as a Pauli (or group of Paulis) -wise compilation process. However, alternative approaches were proposed when Paulis used are in some special form. For example, phase polynomials, i.e. circuits composed only of product of $Z$ and $I$ Paulis can be represented by a parity table~\cite{de2020architecture,vandaele2022phase}, where each row represents a qubit and each column represents a parity whose corresponding angle is not equal to 0. Then the goal is to obtain a parity network~\cite{amy2017cnot}, a circuit that consists of only CNOT gates such that each parity in the parity table appears at least once. With such a parity network, it is enough to place $Z$-rotation gates to implement phase polynomials. \cite{nash2020quantum, kissinger2019cnot, de2020architecture, vandaele2022phase} consider efficient implementation of parity networks employing concepts like Steiner trees to obey the given connectivity. In \cite{meijer2023towards}, the concept is generalized to Pauli strings that consist of either $Z$ or $X$, and mixed ZX-phase polynomials are used for the intermediate representation~\cite{gogioso2022annealing,winderl2023recursively}, tracking parity table in `$Z$ space' and `$X$ space'. These compilation techniques disallows interaction between the $Z$ and $X$ parity tables, e.g. a change in the $X$ parity table depends only on the $X$ parity table and selected CNOT. Both methods present a way of `lazy synthesis' of Pauli networks, in which resetting of the quantum state is not done after some Pauli operators are implemented, but the implementation continues from the current stage of the quantum state, in this context parity tables.

Extending the idea of constructing a parity network where each parity appears at least once to the general Pauli strings, another line of research focuses on efficient construction of Pauli networks. Pauli networks are a generalization of parity networks, where the goal is to produce a Clifford circuit so that each Pauli can be implemented on a single qubit at some point during the execution of the Clifford circuit. In \cite{de2024faster}, authors propose greedy heuristics to minimize either the gate count or depth of entangling gates to implement a given sequence of Pauli strings. Their method however focuses on implementing operations in groups consisting of commuting Pauli strings, and while a method for an ordered sequence of Paulis was proposed, this was done by reducing the problem to the identification of commuting groups at the beginning of the implementation procedure. In \cite{liu2024quclear}, Pauli strings are divided into commuting blocks, and a Clifford circuit is generated recursively to implement the Paulis in each block. Neither of the mentioned works accounts for limited connectivity, however, the first one was recently generalized to limited connectivity~\cite{dilkes2024greedy}. In \cite{martiel2022architecture}, the authors propose lazy synthesis, in which the circuit is built by diagonalizing each component of the Pauli string through a Clifford circuit and a sequence of CNOTs applied along a Steiner tree, which is used to make the procedure architecture aware. 

The papers that avoid immediate resetting of the quantum circuit often rely on some form of an intermediate representation that enables an easier analysis and manipulation of the quantum circuit. Some examples include parity matrix \cite{kissinger2019cnot}, parity table \cite{vandaele2022phase}, parity map \cite{winderl2023recursively}, phase polynomial matrix \cite{meijer2023towards}, tableau \cite{van2020circuit, martiel2022architecture}, Pauli IR program \cite{li2022paulihedral}. Such intermediate presentation is modified with Clifford operations that are added to the circuit. This can be used to track information on the effect of implementing a particular one-qubit Pauli rotation, namely the non-local Pauli exponentiation it corresponds to. 

In this paper, we extend the so far used intermediate representations into Clifford Executive Representation (CER). CER stores the information about the Pauli exponentiation that can be implemented with single-qubit rotation gates at a given moment in the circuit. Using CER, we introduce a generalization of the known compilation techniques where we consider an implementation of multiple, not necessarily commuting Pauli exponentiations at once. Our method comprises a compression component where we move the information scattered among multiple qubits to a small number of qubits and an implementation component where the exponentiations are actually being implemented.  To our knowledge, this is the first compilation method that targets the implementation of multiple, not necessarily commuting Paulis simultaneously and by design works for limited connectivity. Built on this framework, we also present a new Clifford synthesis algorithm. The designed algorithm targets in particular the Clifford circuit that naturally appears at the end of the circuit as a result of the lazy synthesis compilation method. We show that our compilation and Clifford resynthesis techniques outperform the state-of-the-art methods in the literature~\cite{winderl2023architecture} dedicated to limited connectivity. Finally, we demonstrate the applicability of our methods through various examples for VQE. While our algorithm is designed for Pauli networks, its applicability goes beyond as most of the other types of circuits can be converted into Pauli networks.

The paper is organized as follows. In Sec.~\ref{sec:ecr} we present basics on our 
Clifford Executive Representation. In Sec.~\ref{sec:mp-transpilers} we present our family of compilation methods, Clifford resynthesis methods, and benchmarks of our algorithms on ADAPT-VQE circuits for molecular Hamiltonians. Finally in Sec.~\ref{sec:discussion} we conclude and discuss our results.

\section{Clifford Executive Representation} \label{sec:ecr}

Clifford Executive Representation (CER) is a useful formalism to describe the state of the quantum circuit during the computation. It eases the process of compilation and optimization of quantum circuits compared to working with the original quantum circuit. In this section, we start by introducing our intermediate representation for Clifford circuits, which will be the backbone of our compilation algorithm. 

Let there be an $N$-qubit system. Let $\mathbb P_N$ be set of all $N$-long Pauli strings, possibly with a minus sign. Clifford Executive Representation is a function $\ir : \{1,\dots, N\} \times \{Z,X,Y\} \to \mathbb {P}_N$ that satisfies the following:
\begin{enumerate}
    \item $\{\ir_{q,Z}\}_{q}$ and  $\{\ir_{q,X}\}_{q}$  form a valid stabilizers and destabilizers sets,
    \item $\ir_{q,Y} = \ir_{q,Z} \cdot \ir_{q,X}$,
\end{enumerate}
The interpretation of the $\ir$ can be made as follows. 
Suppose that a quantum circuit consists of a Clifford circuit $\mathcal{C}$, a $P$-rotation on qubit $q$ with angle $\theta$ where $P \in \{Z,X,Y\}$ and $C^\dagger$. Let $\ir$ be the corresponding CER obtained from $\mathcal{C}$ (the method for obtaining the corresponding $\ir$ given a Clifford circuit will be explained later). Then the quantum circuit implements the Pauli rotation $ \exp(i\theta\ir_{q,P})$.

Such intermediate representation is insightful as one can look ahead for the Pauli strings that can be implemented at the moment with just single-qubit rotation gates. In~\cite{martiel2022architecture}, a very similar representation called tableau representation but without $\ir_{q,Y}$, namely tableau representation was used for describing the lazy synthesis procedure; we refer readers to Sec. 5 for formal introduction and Appendix C.2 for a particular example of lazy synthesis using tableau representation. Notably, the authors prove that the elements of $\{\ir_{q,Z}\}_{q} \cup \{\ir_{q,X}\}_{q}$ are, in fact, the stabilizers and destabilizers of the inverse of the Clifford circuit that generates the $\ir$. In this paper we go a step further: we use this representation as a strategic resource on how to pursue the compilation. In particular, using such representation we will generalize two approaches to the compilation of the Pauli networks by considering multiple, non-commuting Paulis simultaneously.

In the rest of this section, we will discuss some of the basic facts on $\ir$ and explain how to obtain the corresponding $\ir$ given a Clifford circuit. Let us start by giving some definitions. For a fixed qubit $q$, $P$-register stores the Pauli string $\ir_{q,P}$ for $P \in \{X,Y,Z\}$. For an $N$-qubit system, by $Z_q$, we denote the 1-local $Z$ operator acting on $q$-th qubit, similarly for $X_q$ and $Y_q$.  In addition, 
\begin{itemize}
    \item For any qubit $q$, the Pauli strings stored in the $Y$-register can be directly derived from $Z$- and $X$-registers.
    \item Intermediate representation of empty Clifford circuit satisfies $\ir_{q,P} = P_q$.
\end{itemize}
The first property is true by the very definition of $\ir$. This is also in agreement with the usual Aaronson-Gottesman~\cite{aaronson2004improved} representation of the Clifford circuits, where $2N$ Paulis are sufficient to represent a Clifford circuit. In our case, we decided to keep the $Y$-register Pauli strings for an easier explanation of our compilation techniques, although, in practical implementation, it is more effective to keep track only of $Z$ and $X$ registers. The second property can be directly obtained by observing that a $P$-rotation by angle $\theta$ on qubit $q$ implemented on an empty circuit is equivalent to $ \exp(-i\theta P_q)$. 

Next, we will describe how to obtain the corresponding $\ir$ given a Clifford circuit, which is also in line with~\cite{aaronson2004improved}. Since any Clifford circuit can be implemented as a composition of Clifford gates, i.e. Hadamard, $S$, and CNOT, it is enough to provide rules for these three gates: the rules for other gates can be determined by first decomposing the Clifford circuit into the aforementioned 3 basic gates, and then composing the effect of the gates. The $\ir$ modification rules for the 3 basic gates are presented in Table~\ref{tab:ir_atomic_gates}. As one can see, 1-qubit gates are not changing the Pauli strings themselves except of possible permutation of those among the registers and/or changing the sign of those. On the other hand, CNOT gate introduces 4 new Paulis. 

\begin{table}[t]
    \centering \small
    \begin{tabular}{@{}cc@{}}
 \multicolumn{2}{c}{Hadamard($q$)}\\
    \toprule
          Register&$q$\\
          \midrule 
          $Z$ &$\ir_{q,X}$\\
          $X$&$\ir_{q,Z}$\\
          $Y$&$ -\ir_{q,Y}$\\
          \bottomrule
    \end{tabular}
    \hspace{.4cm} 
      \begin{tabular}{@{}cc@{}}
 \multicolumn{2}{c}{S($q$)}\\
    \toprule
          Register&  $q$\\
          \midrule 
          $Z$ &  $\ir_{q,Z}$\\
          $X$&  $ -\ir_{q,Y}$\\
          $Y$&  $\ir_{q,X}$\\
          \bottomrule
    \end{tabular}    
    \hspace{.4cm} 
    \begin{tabular}{@{}cll@{}}
 \multicolumn{3}{c}{CNOT($c$,$t$)}\\
        \toprule
          Register &$c$&  $t$\\
          \midrule
          $Z$&$\ir_{c,Z}$&  $\ir_{t,Z}\ir_{c,Z}$\\
          $X$ &$\ir_{c,X}\ir_{t,X}$&  $ \ir_{t,X}$\\
          $Y$&$ \ir_{c,Y} \ir_{t,X}$&  $\ir_{t,Y}\ir_{c,Z}$\\
          \bottomrule
    \end{tabular}
    \caption{The new content of the registers after applying Hadamard, $S$ acting on qubit $q$, and CNOT for  control qubit $c$ and target qubit $t$. $\ir$ of the other qubits do not change. }
    \label{tab:ir_atomic_gates}
\end{table}

One can use these rules to provide some simple bounds for the number of CNOTs required for a set of Paulis to be implemented. Let us consider an $N$-qubit,  Clifford circuit with $K$ CNOTs. We are asking how many different Paulis $M$ can appear in the intermediate representations. Taking into account that each CNOT application produces at most 4 Pauli operators that did not exist before, and using the fact that initial intermediate representation already has $3N$ Pauli strings, we can derive a simple bound:
\begin{equation}
        M \leq 3N + 4K.
\end{equation}
Now let us state the following problem: how many CNOTs we need to implement $M$ Pauli strings, assuming the circuit will be in the form of single-qubit rotations each implementing a particular Pauli rotation, sandwiched by Cliford operations? From above we can clearly see that $K \geq \frac{M-3N}{4}$. This directly implies that e.g. UCCSD ansatze requires at least $\Omega(N^4)$ CNOTs to be implemented (as usually expected).

\subsection{Special-purpose Clifford circuits} \label{sec:clifford-circuits}

In the next section, we will explain the usefulness of CER for compilation processes. In there, we will heavily be using so-called Clifford databases composed of small Clifford circuits. Without putting a context in the compilation yet, we define here a few of the special-purpose Clifford circuit types that can be easily explained with CER:
\begin{enumerate}
\item \textbf{Implementation Clifford circuits} Let us assume we want to implement a Pauli network for the ordered set of $l$ Pauli operators $S_1,\dots, S_l$. Implementation Clifford circuits are the ones, for which after applying each gate from the Clifford circuit on the initial CER, it will be modified to contain $S_i$ at some point (taking into account the order of Paulis, but also the commutativity relation among the Paulis). Note that since as we showed before 1-qubit gates are not creating new Pauli operators in CER, it is enough to investigate qubits that were altered by 2-qubit operations like CNOTs.

\item \textbf{Simultaneous implementation Clifford circuits} Let us assume we want to implement a Pauli network for the ordered set of $l$ Pauli operators $S_1,\dots, S_l$ with only 1-qubit rotations. Simultaneous implementation Clifford circuits construct CER such that \textit{all} Paulis appear on $Z$, $X$, or $Y$ registers of some qubits. Consequently, after applying the Clifford operation, all Pauli can be implemented with just $Z$-, $X$- or $Y$-rotations.  Note that Pauli operations $S_1,\dots, S_l$ need to satisfy special requirements: They need to be partitioned into sets of at most 3 elements, s.t. Paulis from in between the sets commute, and within the sets does not commute. In addition, for 3-element sets $\{S', S'', S'''\}$ a relation $S''' \propto S' \cdot S''$ must hold due to the very definition of CER.

\item \textbf{Compressing Clifford circuits} For initial CER over qubits $q_1,\dots, q_{N}$, we want to compress Paulis $S_1,\dots, S_k$ into qubits $q_2,\dots, q_{N}$. The compressing Clifford circuits transform the initial CER into $\ir$ s.t. for all $i=1,\dots, k$
    \begin{equation}
        S_i = \pm \prod_{q=2}^N \ir_{q,Z} ^{\varepsilon_{q,Z,l}}\ir_{q,X} ^{\varepsilon_{q,X,l}}\ir_{q,Y} ^{\varepsilon_{q,Y,l}}, \qquad l=1,\dots, k,
    \end{equation}
where $\varepsilon_{q,P,l} \in \{0,1\}$ for $P \in \{Z,X,Y\}$ and $\sum_{P\in X,Y,Z}\varepsilon_{q,P,l}\leq 1$. Note that qubit $q_1$ is omitted in the above product.
\end{enumerate}

\section{Multi-Pauli compilation}\label{sec:mp-transpilers}

\subsection{Pauli networks}

In this section, we present a method that utilizes the previously introduced $\ir$ to provide an extension to the state-of-the-art compilation methods. For now, we will assume that our circuit is in the form of a sequence of single-qubit Pauli rotations, i.e. state to synthesize is of the form
\begin{equation}
    \ket{\psi} = \prod_{l=1}^L \exp(i \theta_l S_l),
\end{equation}
where $\theta_l$ is a real parameter and $S_l \in \mathbb P_N$ is a Pauli string. Later, we will show how our results can be applied to general circuits. 

Our method works by taking the first $k$ Pauli strings in the above circuit and applying the Clifford circuit that localizes them in some subset of $n \ll N$ qubits, so-called the \textit{compression component}. Eventually, once the Paulis are localized, they are being implemented with an \textit{implementation component}. Optionally, a \textit{resetting component} is used in which the initial $\ir$ is recovered. The whole process aims at transforming $\ir$ in a way so that $\ir_{q,P} = (-1)^{b_l} S_l$ at some point for some qubit $q$, Pauli $P$ at some point, and $b_l \in \{0,1\}$. Then implementing single-qubit rotation $\exp(i\theta_l(-1)^{b_l}P_q)$ will be  equivalent to implementing $\exp(i \theta_l S_l)$. 

A particular, yet already quite general procedure goes as follows. First, we take the $k$ Pauli strings $\mathcal S = (S_1,\dots,S_k)$ to be implemented from $\ket{\psi}$. Since Paulis in the $Z$- and $X$-registers of $\ir$ are algebraically independent, one can find a decomposition for each of the Paulis $S_1,\dots, S_k$ as follows:
\begin{equation}
    S_l = \pm \prod_{q=1}^N \ir_{q,Z} ^{\varepsilon_{q,Z,l}}\ir_{q,X} ^{\varepsilon_{q,X,l}}\ir_{q,Y} ^{\varepsilon_{q,Y,l}}, \qquad l=1,\dots, k,
\end{equation}
where $\varepsilon_{q,P,l} \in \{0,1\}$ for $P \in \{Z,X,Y\}$ and $\sum_{P\in X,Y,Z}\varepsilon_{q,P,l}\leq 1$. Next, we take a set of qubits $Q$ over which at least one Pauli $S_l\in \mathcal S$ is acting nontrivially, i.e. for which $\sum_{P\in\{X,Y,Z\}}\varepsilon_{q,P,l}=1$ and we take a subgraph of the given quantum hardware graph that contains all the aforementioned nodes. A usual good choice of such a subgraph is the induced subgraph over nodes belonging to the Steiner tree over the nodes $Q$. Once having such a graph, we begin the compression in which we remove information qubit by qubit from the subgraph, eventually decreasing the number of qubits over which Paulis $S_l\in \mathcal S$ are defined. This process for small $k$ can be done efficiently by using the database of compression Clifford circuits. The details on how to construct Clifford databases are described in Appendix~\ref{sec:clifford-database}. The process of compressing the information ends once the subgraph has some pre-chosen number of qubits $n\ll N$. Once this is achieved, implementation components are added so that Paulis $S_l \in \mathcal S$ are being implemented. For this phase, any method that effectively compiles the Pauli networks over a small number of qubits for limited connectivity can be chosen, although, for small $k$ and $n$, we again found out that a pre-computed implementation Clifford database gives the best results. Eventually, an optional resetting phase is implemented and the same procedure is repeated for the next list of Paulis from $\ket{\psi}$. A visualization of this procedure can be found in Fig.~\ref{fig:example}, and the final circuit can be found in Fig~\ref{fig:example-result}.

\begin{figure}
    \centering
    \includegraphics[scale=0.75]{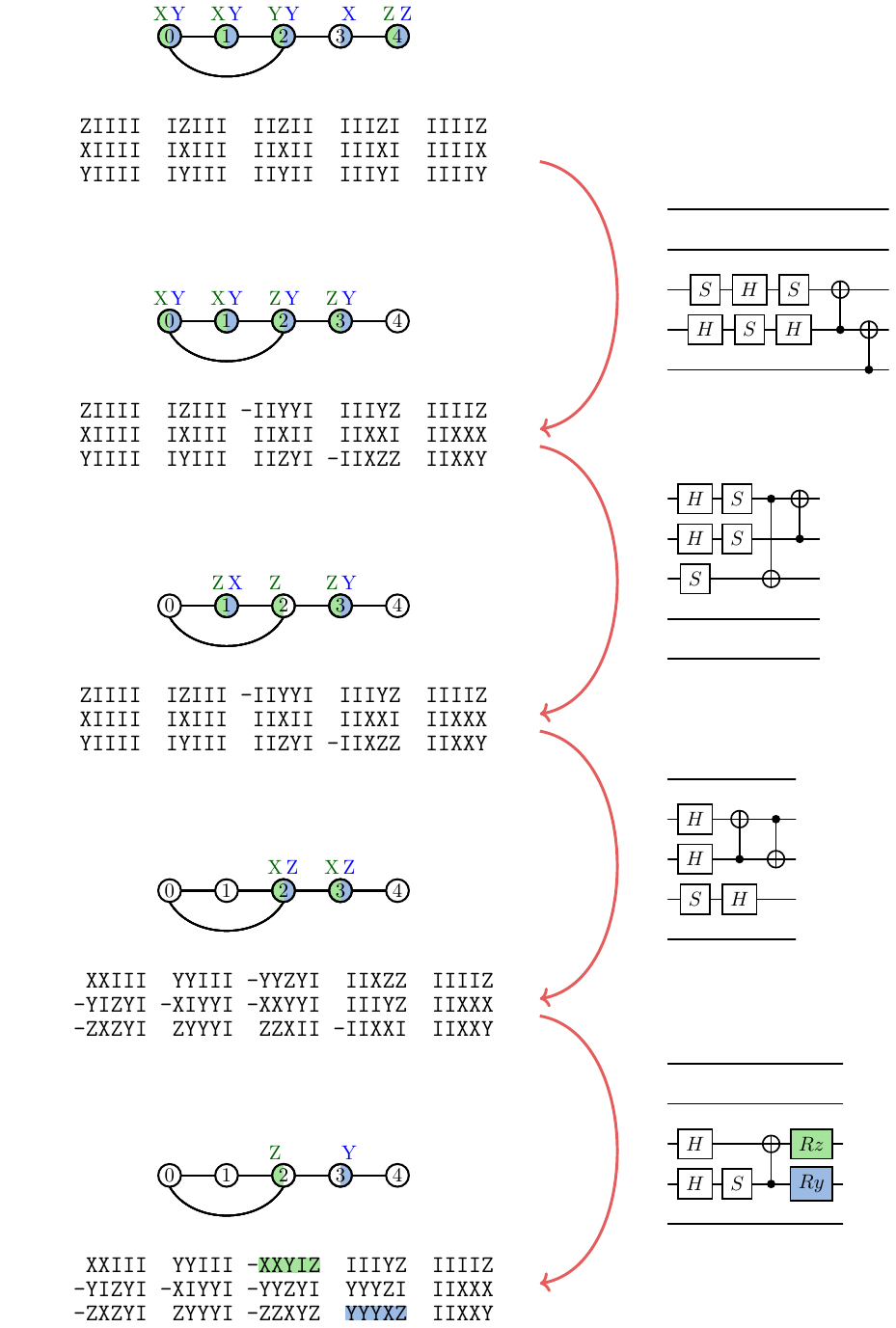}
    \caption{An example of simple routine of compiling $\exp(i\theta_1XXYIZ)\exp(i\theta_1YYYXZ)$. We start with the initial intermediate representation. Via 3 compression circuits, we compress the information to qubits 2 and 3. Eventually, the implementation circuit implements Pauli $XXYIZ$ appears on the $Z$ register of qubit 2, and $YYYXZ$ on the $Y$ register of qubit 3. Not that the selected Clifford operations take into account additional connections that are usually discarded when the Steiner tree is considered. Note that one can use special Clifford operations as described in Sec.~\ref{sec:clifford-circuits}, for example, the second circuit compresses $XXZ$ and $YYY$ (based on in which register the information is stored) from qubits 0, 1, 2 to qubits 1, 2. The last Clifford circuit is an implementation circuit for Paulis $XX$ and $ZZ$.}
    \label{fig:example}
\end{figure}

\begin{figure}
    \centering
    \includegraphics[width=\textwidth]{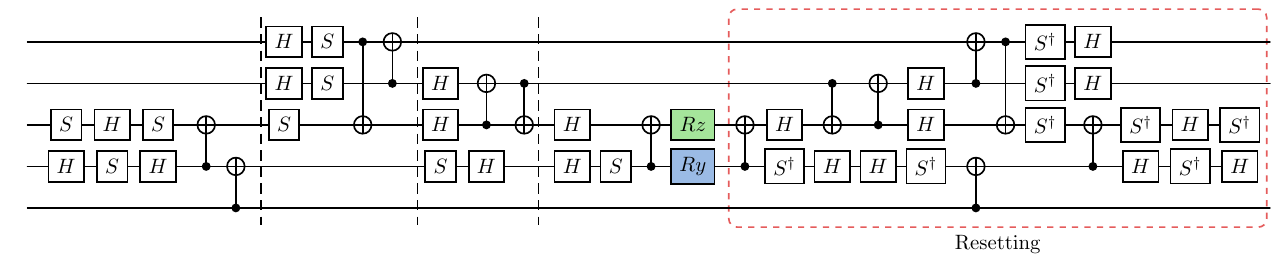}
    \caption{Final circuit from the procedure presented in Fig.~\ref{fig:example}, with an additional resetting component. Horizontal dashed lines separate components generated in the procedure.}
    \label{fig:example-result}
\end{figure}

An interesting question is what is the minimal $n$ to which information can be compressed. Clearly, $n$ cannot be arbitrarily small, as over a subset of qubits $Q$, one can implement at most $2^{2|Q|}-1$ Pauli strings of the form
\begin{equation}
    \prod_{q\in Q} \ir_{q,Z} ^{\varepsilon_{q,Z}}\ir_{q,X} ^{\varepsilon_{q,X}},
\end{equation}
where we assume that at least one $\varepsilon_{q,P}$ is nonzero, as otherwise, we get identity Pauli which affects only the unphysical global phase. However, not only the number of Pauli strings is important, but also their commutativity structure. For example, if we take two anti-commuting Paulis $P, P'$, one can find $\ir$ s.t. for some qubit $q$ we would have e.g. $\ir_{q,Z} = \pm P'$ and $\ir_{q,X} = \pm P''$. In addition, we can `squeeze' even a third Pauli $P' \cdot P''$ from the $Y$ register, effectively allowing compression of 3 Paulis into a single qubit. However, two commuting Paulis cannot be placed in a single qubit for any $\ir$ by the very definition, so at least two qubits are needed to store Paulis in e.g. $\ir_{q,Z}$ and $\ir_{q',Z}$ for $q\neq q'$. For these reasons, and also due to the fact that locally compression components as explained above may encounter commuting Paulis irrespective of the commutation structure of $S_1,\dots,S_k$, we will be always compressing $k$ Paulis into $n=k$ qubits, which is always possible, see Appendix~\ref{sec:compression}.

Let us point out that the compression and implementation components above in fact produce a large family of methods. One can point out the following degrees of freedom among others:
\begin{enumerate}
    \item The compression and implementation circuits might be chosen to optimize CNOT count, depth, or combination of those, which allows fine-tuning Clifford databases towards specific needs.
    
    \item When removing a qubit, the qubit may be chosen randomly, or for instance, the one whose removal results in the smallest number of CNOTs may be chosen.
    
    \item If the next Pauli to be implemented is encountered in the $\ir$ before the compression into $n$ qubits, it can be implemented before the compression is complete. Commutativity among the Pauli strings can be also considered to determine the order of implementation. In this case, a Pauli among the list can be implemented as soon as it is encountered if it commutes with the preceding ones. One can generalize this rule to the case in which the Pauli strings are decomposed into a few registers over a small connected subgraph.
    
    \item One may prefer to implement a larger number of Paulis simultaneously so that each Pauli appears in some register. Although this may require a more complicated implementation circuit, this might be practical for the fault-tolerant era as this will reduce T gate depth, which results only due to Pauli rotations in our case. Note that Pauli operators need to satisfy special criteria specified by the definition of CER to make it possible.
    
    \item Instead of taking the next $k$ Paulis, one can consider a more advanced selection. Instead of fixing the number of the next Paulis to implement, one can consider implementing $i$ Paulis for $i = 1,\dots,k$ and check the total CNOT count w.r.t. to the number of Paulis implemented for each case. Note that $i=k$ is the case described above. One can choose $i'$ that minimizes the CNOT count and implement $i'$ next Pauli strings in the sequence. One can additionally take into account commutativity when selecting the next Paulis. 
    \item To boost compression, one can relax the condition of taking $k$ Paulis to take $k' \geq k$ Paulis s.t. that all $k'$ Paulis belongs to the space spanned by some $k$ Paulis. Note that for compression we only need to take into account those $k$ Paulis, and the remaining $k'-k$ will be compressed alongside.
\end{enumerate}
We would like to emphasize that the above list does not include all the possible degrees of freedom and other possible advancements can be considered. 

While the order of implementation and compression components can be very complex, we found the particular workflow described below particularly useful. First, we specify the list of components that are supposed to implement a list of Paulis of length at most $k'$. Such components will be executed in a loop until all the Paulis are implemented. Then, a list of consecutive Paulis $S_1,\dots,S_{k'}$ is selected so that its rank is at most preselected number $k\leq k'$. Then for the sublists $P_1,\dots, P_i$ for $i=1,\dots, k'$ we repeat inside the loop the previously defined list of components. After we get the circuits combining compression, implementation, and resetting (if applied) for all the sublists of size $i$, we choose the one for which the number of CNOTs over $i$ is the smallest.

Right after implementing the compression and implementation components, it is likely that we end up with an $\ir$ different than the initial one that corresponds to the empty circuit. Two different approaches from the literature can be considered. The first one is to implement the resetting right after implementing each sublist. In the second, resetting is done only once all the Paulis are implemented. Resetting $\ir$ can be done by taking all so far-produced gates (that were not reset before), removing the non-Clifford ones responsible for implementing Paulis, and applying the inverse of them at the end of the circuit. Such circuits can be optimized using Clifford circuit resynthesis methods like the one proposed in \cite{winderl2023recursively} or using the methods that will be presented in the next section. Note that in the case where resetting is done after each sublist generation, not much improvement may be expected for a large number of qubits, as this would mean that the compression and implementation phases themselves are not done efficiently. However, if resetting was not done for the many lists of Paulis, a much larger improvement is expected. Alternatively, if an expectation value of some Hamiltonian will be calculated with respect to the prepared state, the Clifford part can be merged to the Hamiltonian as already suggested in \cite{liu2024quclear}.

Using the multi-Pauli compilation framework described above, we obtain a compilation algorithm which we name Multi-Pauli Lazy Synthesis (MPLS). We will compare it against lazy synthesis~\cite{martiel2022architecture} (LS) and Steiner synthesis (SS), an adaptation of phase Pauli synthesis using Steiner trees to general Pauli strings as explained in~\cite{miller2024treespilation}. For LS and MPLS, resetting is executed at the end of the circuit only. We considered circuits with varying numbers of Paulis and qubits, where Paulis correspond to products of 2 or 4 Majorana operators using the Jordan-Wigner (JW)~\cite{jordan1993paulische} and Bravyi-Kitaev (BK)~\cite{bravyi2002fermionic} mappings. The choice of Pauli strings is dictated by the ADAPT-VQE procedure as explained in~\cite{miller2024treespilation}. We compiled the circuit for the path graph, with layout chosen for each node of respective PPTT ternary tree~\cite{miller2024treespilation} being the same for qubit and mode, i.e. if the node in a tree was assigned to $i$-th mode, it is also assigned to $i$-th qubit. The results are presented in Figs.~\ref{fig:per-paulis} and \ref{fig:per-qubits}, and technical details can be found in Appendix~\ref{sec:numerics-general}. We can see that our method consistently outperforms the other two methods for both mappings. The improvement over LS remains when we consider both methods without the resetting part as if it would be `merged' into a Hamiltonian. 

\begin{figure}[t]
    \centering
    \includegraphics[scale=0.8]{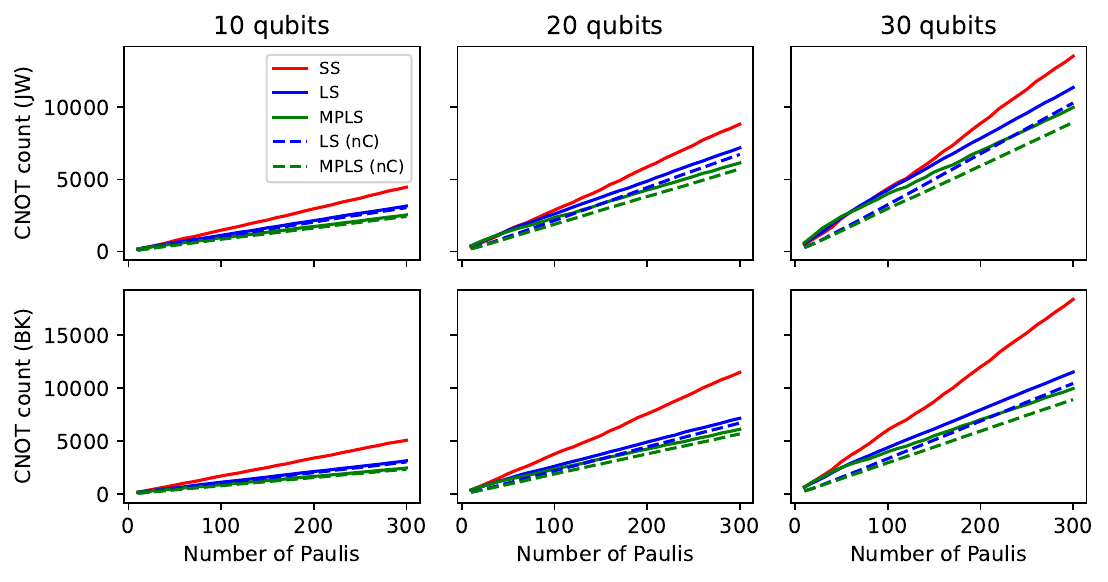}
        \caption{CNOT counts for various numbers of qubits and an increasing number of Pauli exponentiations, obtained by lazy synthesis (LS), Steiner synthesis (SS), and Multi-Pauli Lazy Synthesis (MPLS) methods. JW mapping is used for the top plots and BK mapping is used for the bottom plots. For LS and MPLS, CNOT counts without the final Clifford circuit are also presented.}
    \label{fig:per-paulis}
\end{figure}
\begin{figure}[t]
    \centering
    \includegraphics[scale=0.8]{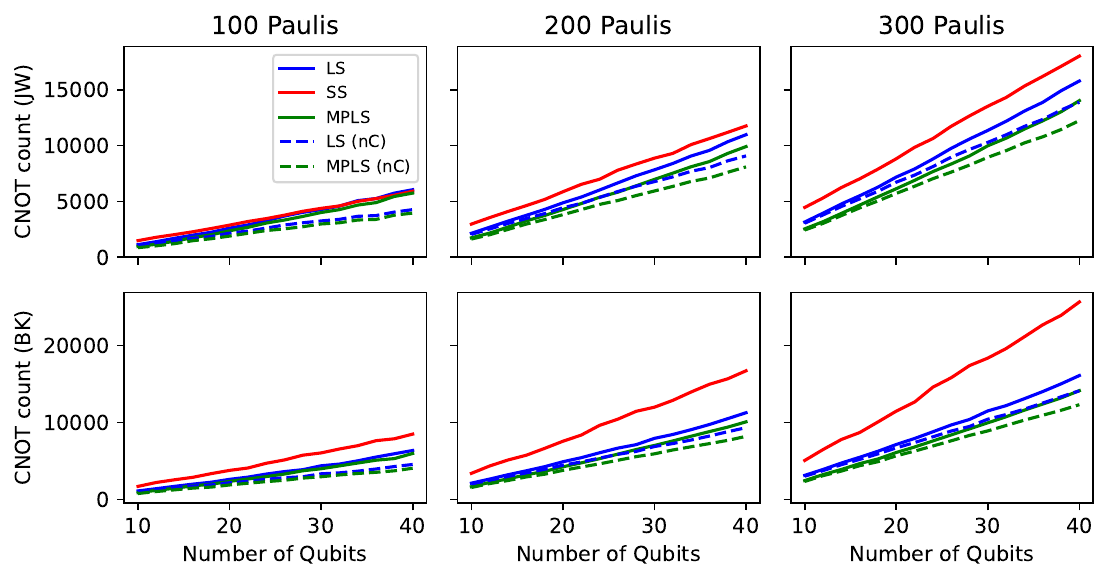}
    \caption{CNOT counts for various numbers of Pauli exponentiations and an increasing number of qubits, obtained by lazy synthesis (LS), Steiner synthesis (SS), and Multi-Pauli Lazy Synthesis (MPLS) methods. JW mapping is used for the top plots and BK mapping is used for the bottom plots. For LS and MPLS, CNOT counts without the final Clifford circuit are also presented.}
    \label{fig:per-qubits}
\end{figure}

When the resetting is left to the end of the circuit, the intermediate representation eventually becomes more and more complicated. Starting from an arbitrary complicated $\ir$, information from $\Theta(n)$ qubits is required to be able to implement most of the next Paulis. While this makes the method still efficient for the architectures in which the average distance between the nodes is $\Theta(n)$ like a path graph, or for the Pauli strings which are on average already $\Theta(n)$-local, recent studies on the fermion-to-qubit mappings show that significantly smaller locality can be guaranteed for the ansatze consisting of Pauli operators being low-order products of Majorana strings. In particular, for heavy-hexagonal architectures or 2D grids, a properly chosen mapping can guarantee a locality of at most $\order{\sqrt n }$. This means that the proposed method will not be efficient, as the Steiner synthesis already guarantees $\order{M\sqrt n}$ 2-qubit gates for $M$ Pauli strings, while the proposed method will likely require $\order{Mn}$. For such cases, it is recommended to reset after each of the Paulis from a list is implemented. 

To demonstrate that our method is beneficial also with immediate resetting, we considered fermionic spin operators from the original ADAPT-VQE~\cite{grimsley2019adaptive}, where the generators of double excitations take the form
\begin{equation} 
\begin{split}
\quad a_i^\dagger a_j^\dagger a_ka_l - a_k^\dagger a_l^\dagger a_ia_j &= -\frac{i}{8}(m_im_jm_k\bar m_l + m_im_j \bar m_k m_l - m_i \bar m_j  m_k m_l - \bar m_im_j  m_k m_l \\
&\phantom{=\ } - \bar m_i \bar m_j \bar m_k m_l - \bar m_i \bar m_j  m_k  \bar m_l + \bar m_im_j \bar m_k \bar m_l + m_i \bar m_j \bar m_k  \bar m_l),
\end{split}
\end{equation}
for some  pairwise different mode indices $i,j,k,l$, where we use the standard equivalence $m_j = a_j^\dagger + a_j$ and $\bar m_j = i(a_j^\dagger-a_j)$. As one can see, the first four elements are enough to produce all the terms in the Majoranic polynomial, e.g.
\begin{equation}
    \bar m_i \bar m_j \bar m_k m_l \propto m_im_jm_k\bar m_l \times m_i \bar m_j  m_k m_l \times m_im_j  m_k m_l,
\end{equation}
which similarly follows for other Majoranic monomials from the second row. Hence it is enough to compress these 4 Majorana products into 4 qubits, and then implement a 4-qubit Clifford circuit which produces a sequence of intermediate representations such that every term appears at least once. Note that since all the terms commute, the order can be arbitrary. We benchmark LS and SS against the Multi-Pauli Resetted (MPR) variant for random circuits with a single double-excitation. For the layout and the mapping, we take a random PPTT with the corresponding ternary tree with the minimum height from a heavy-hexagonal architecture. The results are presented in Fig.~\ref{fig:mpp-4}, and technical details can be found in Appendix~\ref{sec:numerics-de}. As we can see, our method clearly outperforms the SS, and also provides better circuits than LS in many cases, which suggests that the combination of MPR and LS might be preferable in practical scenarios.

\begin{figure}[t]
    \centering
    \includegraphics[scale=0.8]{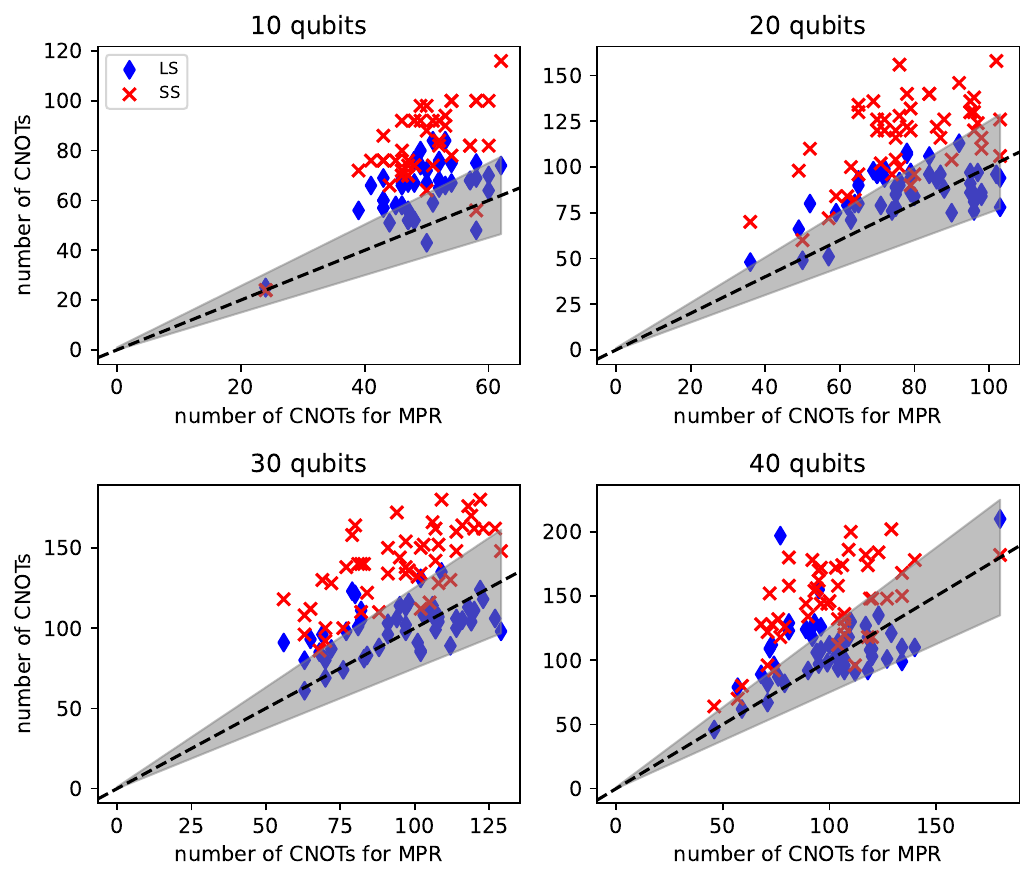}
    \caption{CNOT counts for random circuits that consist of a single double excitation. The dashed line depicts the line $x=y$. The gray area spans the space between $y=0.75x$ and $y=1.25x$.}
    \label{fig:mpp-4}
\end{figure}

While the proposed compilation method originally targets Pauli networks, its applicability goes beyond these types of circuits. This is due to the fact that most of the circuits can be easily converted into a Pauli network followed by a Clifford circuit. In particular, if we are given a quantum circuit consisting only of Clifford+$T$ gates, the first $T$ gates can be converted into $Z$-rotations. Then, if we have a sequence of Clifford gates and Pauli rotations, one can transform them easily into a sequence of Pauli network followed by a single large Clifford circuit that could be resynthesized~\cite{Litinski2019magicstate}.

\subsection{Multi-Pauli Clifford Synthesis}\label{sec:mlcs}

The compilation procedure described in the previous subsection for Pauli networks can be transformed into a Clifford synthesis method. Suppose we are given a Clifford operation $C$. In the scope of this paper, the task of Clifford synthesis is to create a quantum circuit obeying limited connectivity constraints that implements $C$. We would like to point out that there exists methods of Clifford synthesis directly designed for linear nearest neighbor architecture~\cite{bravyi2021hadamard} with appealing $\order{n}$ depth, however, its generalization to arbitrary connectivity is not known. Clifford synthesis algorithm from~\cite{winderl2023architecture} steps forward as an algorithm designed for limited architecture. Recently another method was proposed that was synthesizing Clifford circuits in the pairs of $Z_q,X_q$ simultaneously (qubit by qubit)~\cite{dilkes2024greedy}, instead of first $Z_q$ and then $X_q$ (Pauli by Pauli). This is in the spirit of the general compilation method presented in the previous subsection, however,  the synthesis was designed by a list of non-straightforward rules on how to `compress' information from $Z_i$ and $X_i$ into a qubit.

In this paper, we propose a generalization of the methods presented. In our case, similarly, as it was done in~\cite{dilkes2024greedy} the concept is to compress operators $Z_i, X_i$ for a few $i=i_1,\dots,i_k$ ($k=1$ in the original paper). However, instead of relying on the rules specified a priori, we will rely on the Clifford databases used also in the previous step to compress information into a subset of physical qubits, and a particular compression-implementation workflow explained later. Note that while the information has to be compressed to $k$ qubits, the logical input qubits and output qubits might differ, in the sense that the Clifford operation will be synthesized up to a permutation of the logical qubits embedding.

Let us now describe the exact procedure. We start with generating $\ir$ for a given Clifford operation $C$. This can be easily done by first taking an arbitrary Clifford synthesis method (not necessarily for limited connectivity) and constructing $\ir$ using the rules presented in Table~\ref{tab:ir_atomic_gates}. We select a set of \textit{logical} qubits $i=1,\dots, i_k$ and physical qubits $q_{j_1},\dots,q_{j_k}$ that form a connected subgraph and do not disconnect the graph. Then we follow the compression procedure so that the Pauli list $\{Z_{i_1}, \dots, Z_{i_k}, X_{i_1}, \dots, X_{i_k}\}$ will be compressed onto qubits $q_1,\dots, q_k$. Here, we can use any of the strategies presented in the previous section. Once compression is done, an implementation component is executed that generates 1-local $Z$ and $X$ Paulis from the list in the respective physical qubits. Note that in this case simultaneous implementation for all $Z_i, X_i$ is required. Eventually, the qubits $q_1,\dots,q_k$ are removed from the graph, and the procedure repeats for new indices $i$. Note that one can implement a version in which the same logical qubit $i$ is always `uncomputed' to the physical qubit $q_i$ (assuming before the Clifford operation we have trivial embedding $i\mapsto q_i$), however, in this case, it is natural to first specify the subgraph induced by $q_{j_1},\dots, q_{j_k}$ and the do the compression for Pauli list $\{Z_{j_1}, \dots, Z_{j_k}, X_{j_1}, \dots, X_{j_k}\}$. The ambiguity of specifying the final layout of logical to physical qubits might be useful for lazy synthesis methods, in the case Clifford operations are implemented right before the measurement.

It is important to point out noticeable differences between the above procedure and the procedure from~\cite{dilkes2024greedy}. First, our methods directly generalized the procedure into the larger number of qubits being compressed, e.g. $k>1$. Second, our procedure possesses more flexibility in the direction towards the selection of the Clifford databases, or how to perform the compression scheme. As an example, Clifford databases optimized towards depth minimization can be chosen. Finally, the rules presented in~\cite{dilkes2024greedy} were naturally designed for Steiner trees; in our case, the Clifford database could utilize reacher connectivity that may appear in the subgraph during the compression.     

In Fig.~\ref{fig:clifford}, we present benchmarks for our exact and unordered (up to SWAPs on the right side) Multi-Pauli Clifford Synthesis (MPCS) with 2 Pauli compression at a time, against the method from~\cite{winderl2023architecture} which we will refer to as PauliOpt, and an improved unordered version of PauliOpt in which similarly as in the method above logical qubits can be reset to different qubits. We considered Clifford circuits with varying numbers of qubits and gates, technical details for circuit generation to be found in Appendix~\ref{sec:numerics-clifford}. For the underlying architecture, we use a path graph. We can see that our methods consistently improve over the existing state-of-the-art methods. Because of the lack of implementation of the method presented in~\cite{dilkes2024greedy}, we could not benchmark against it, but we expect the numbers to be comparable to ours, as one can consider this method to be a special version of the proposed one.

\begin{figure}[h]
        \centering
        \includegraphics[scale=0.8]{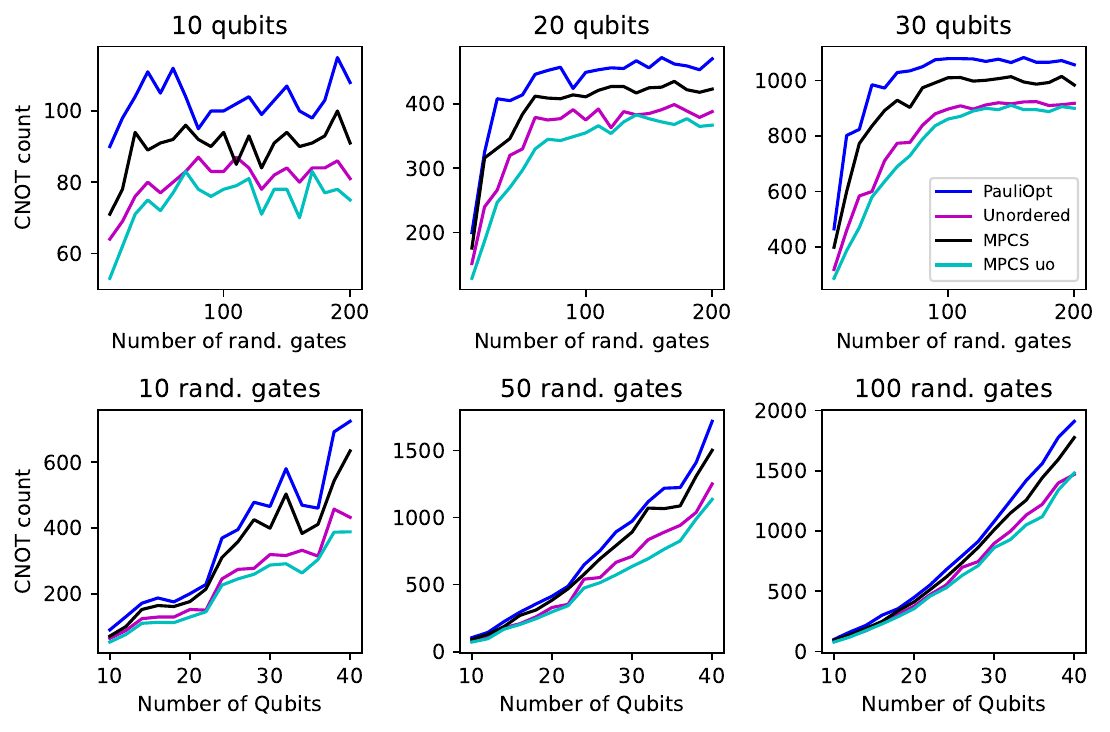}
        \caption{CNOT counts for different Clifford synthesis methods for increasing number of gates and fixed number of qubits (top) and for a fixed number of gates and increasing number of qubits (bottom). `MPCS' and `MPCS uo' stand for exact and unordered Multi-Pauli Clifford Synthesis respectively. `Unordered' stands for the unordered version of the PauliOpt algorithm~\cite{winderl2023architecture}.}
        \label{fig:clifford}
    \end{figure}

\subsection{Benchmarks for VQE-generated ansatze}

For simulating Fermionic systems, as in the case of VQE, it is crucial to map the problem to the qubit space by using a Fermion-to-qubit (F2Q) mapping. In \cite{miller2024treespilation}, authors present a technique to reduce CNOT count on limited connectivity by optimizing over a family of tree-based mappings. During the optimization process, a cost function is chosen to reflect the CNOT count of the resulting circuit. Once the F2Q mapping is chosen, ansatz is compiled using a particular compilation method in accordance with the cost function chosen during optimization. Using our compilation method as a subroutine, we compared treespilation with LS and SS vs treespilation with MPLS as described in previous section for ground state ansatze prepared by ADAPT-VQE algorithm using the majoranic pool for BODIPY-4 molecule as defined in~\cite{bodipy}, H$_6$, LiH, N$_2$, BeH$_2$ molecules as considered in~\cite{miller2024treespilation} on the heavy-hexagonal architecture of IBM Eagle. We used the PauliOpt algorithm and multi-Pauli Clifford synthesis algorithms for LS and MPLS, respectively, to resynthesize the Clifford circuit at the end of the ansatz. The results are presented in Table~\ref{tab:treespilation}, technical details to be found in Appendix~\ref{sec:numerics-treespilation}. It is clearly seen that treespilation when combined with MPLS outperforms treespilation with SS and LS. 

\begin{table}[h]
    \centering
    \begin{tabular}{lccc}
    \toprule
         Molecule&  SS&  LS& MPLS\\
         \midrule
         H$_6$ (12)&  2888&  2444& 1869\\
 LiH (12)& 110& 113&100\\
 N$_2$ (12)& 372& 395&295\\
 BeH$_2$ (14)& 866& 919&722\\
 BODIPY-4 (16)& 2140& 2246&1845\\
 \bottomrule
    \end{tabular}
    \caption{CNOT counts of ground state ansatze prepared using ADAPT-VQE majoranic pool compiled for IBM Eagle device using treespilation \cite{miller2024treespilation} combined with different compilation methods. The number inside the parenthesis next to each molecule is the number of qubits. }
    \label{tab:treespilation}
\end{table}

\section{Discussion} \label{sec:discussion}

Given a quantum state in the form of a sequence of  Pauli strings, its efficient compilation is crucial, in particular for architectures with limited connectivity. There have been various methods proposed in the literature, making use of concepts like phase polynomials, parity networks, and Clifford tableaus as mentioned in the Introduction, which represent the current `status' of the quantum state. In this paper, we pursued a similar approach and introduced Clifford Executive Representation, a convenient generalization of tableau representation presented in~\cite{martiel2022architecture}, to serve as our intermediate representation when inspecting such quantum circuits. Employing this representation, we presented a compilation algorithm that implements multiple (not necessarily commuting) Pauli strings at once. 

Our compilation method comprises a compression component that lies at the heart of our synthesis technique. With appropriate compression, the information initially scattered around multiple qubits is localized into a smaller subset. Compression in our case is performed by removing information qubit-by-qubit, using connections beyond Steiner trees usually considered in the literature. The core problem of giving an effective compression is to find an efficient Clifford circuit that removes information from a qubit. We did this by pre-computing databases of Clifford circuits that serve the purpose. The requirements for Clifford databases can be easily explained using CER, which gives yet another application of the representation. While our Clifford database was optimized towards minizing the number of CNOTs, there is flexibility on what an efficient Clifford operation is and one can consider other metrics, e.g. depth of the circuit.

We benchmarked our algorithms against random ansatze, but also ansatze coming from an ADAPT-VQE procedure for molecular Hamiltonians. In all the cases considered, we found out that our method outperforms the state-of-the-art methods designed for limited connectivity. Our method also performs well for various F2Q mappings and outperforms Steiner Synthesis~\cite{miller2024treespilation} and Lazy synthesis methods~\cite{martiel2022architecture} also when combined with treespilation.
Furthermore, we proposed a new Clifford resynthesis method for limited connectivity based on multi-Pauli lazy synthesis, which outperforms the state-of-the-art methods.

Our work can be improved in various directions. In particular, we believe it is worth exploring how the method improves with a larger number of Paulis in the sublist. We were able to reach $k=4$ which was still providing some improvement, however, eventually we found difficulty with generating Clifford databases for $k>4$. Effective pre-computing, or constructing Clifford operations for compression or implementation for many qubits and/or Paulis is another bottleneck of our algorithm. Finally, while our numerical results considered only Pauli lists consisting of consecutive Paulis in the ansatz, we believe more advanced Pauli list selection could be done. 

\paragraph{Data availability}
The data that support the findings of this article are available from the corresponding author upon reasonable request. 

\paragraph{Code availability }
The code that was used to obtain the data that supports the findings in this article cannot be made publicly available as it is proprietary and part of Algorithmiq's Aurora platform.

\paragraph{Acknowledgements}
We thank Zolt\'an Zimbor\'as for helpful discussions. Work on “Quantum Computing for Photon-Drug Interactions in Cancer Prevention and Treatment” is supported by Wellcome Leap as part of the Q4Bio Program.

\paragraph{Author contributions}
A.G. proposed the research idea. Both authors contributed to every other aspect of the research and manuscript preparation.

\bibliographystyle{ieeetr}
\bibliography{sample}

\appendix

\section{Clifford databases} \label{sec:clifford-database}

Our method heavily relies on well-adjusted Clifford databases that locally transform the intermediate representation. The necessary steps for creating such a Clifford database comprise of:
\begin{enumerate}
    \item Formally specifying the requirements of the Clifford operation. In our case, we only use compressing and implementing Cliffords from~\ref{sec:clifford-circuits} 
    \item Formally specifying the constraints of the Clifford operation. This may in particular include information about the pairs of qubits that can be interacted with CNOTs.
    \item In case there is more than one Clifford circuit obeying the specified requirements and constraints, one may optionally add a cost function that assesses the `goodness' of the circuit. In our case, we focus on finding Clifford circuits that minimize the number of CNOT gates.
\end{enumerate}
Our Clifford databases are constructed heuristically as follows. First, we specify a maximum number of CNOTs $K$ -- this is adapted heuristically to make sure that some Clifford circuit of interest is generated. Then starting with an empty circuit, we randomly select two qubits that can be interacted with CNOTs. On both qubits, we apply a random 1-qubit Clifford gate, and then we apply a CNOT gate afterward. This procedure is repeated until $K$ CNOTs are added, and each circuit with $0,\dots,K$ CNOTs is verified to check if it satisfies the necessary requirements and conditions. Such Clifford circuits are generated for multiple numbers of Paulis and nodes, also considering various subgraphs that can appear on a path graph or a heavy-hexagonal graph considered in this paper. $K$ is heuristically chosen so that all the Clifford circuits are produced.

\section{Compression of $M$ Paulis into $M$ qubits} \label{sec:compression}

In here, we adopt the notation $P' \propto P''$ if $P=P'$ or $P=P''$.

\begin{theorem}
For any list of Paulis $P_1,\dots,P_M$ defined over $N > M$ qubits and qubits $Q = \{q_1,\dots, q_M\}$, there exists a Clifford circuit $C$ that localizes them into $Q$, i.e. for all $m=1,\dots,M$ we have
\begin{equation}
    CP_m C^\dagger  \propto \prod_{i=1}^M Z_{q_i}^{b_{m,Z,i}} X_{q_i}^{b_{m,X,i}},
\end{equation}
for some $b_{m,z,i}, b_{m,x,i} \in \{0, 1\}$.
\end{theorem}
\begin{proof}
    We will prove the theorem by induction. The case $M=1$ can be done by using e.g. Steiner Synthesis~\cite{miller2024treespilation} to compress into an arbitrary qubit, and then swapping it with qubit $q_1$.

    Let us assume that $M>1$ and suppose we have a Clifford circuit $C'$ that localizes Paulis $P_1,\dots, P_{M-1}$  into qubits $q_1,\dots, q_{M-1}$. Let $\tilde P$ be defined such that 
    \begin{equation}
        C'P_M(C')^\dagger = \pm \prod_{i=1}^N  Z_{q_i}^{b_{M,Z,i}} X_{q_i}^{b_{M,X,i}} = \tilde P \prod_{i=1}^{M-1}  Z_{q_i}^{b_{M,Z,i}} X_{q_i}^{b_{M,X,i}}. \label{eq:tilde-p}
    \end{equation}
    Note that $\tilde P$ acts with identity on the qubits $q_1,\dots, q_{M-1}$. If $\tilde P \propto I$, then the Paulis are already localized into qubits $q_1,\dots, q_{M-1}$, and thus to $q_1,\dots, q_{M}$. 
    
    Suppose on the contrary that $\tilde P \not\propto I$. But then we can reduce the problem of compressing all $M$ Paulis to compressing $\tilde P$ from qubits $q_{M}\dots, q_N$ into $q_M$ which can be done the same as in the base of the induction, say with Clifford circuit $C''$ that does not act on qubits $q_1,\dots,q_{M-1}$. Eventually, let us show that Clifford circuit $C \coloneqq C''C'$ localizes all the Paulis $P_1,\dots, P_M$ into qubits $q_1,\dots, q_M$. We have
    \begin{equation}
        \begin{split}
        C P_MC^\dagger &= C''C'P_M(C')^\dagger (C'')^\dagger \\
            &= C''  \tilde P \prod_{i=1}^{M-1}  Z_{q_i}^{b_{M,Z,i}} X_{q_i}^{b_{M,X,i}} (C'')^\dagger \\
            & =C''  \tilde P (C'')^\dagger \prod_{i=1}^{M-1}  Z_{q_i}^{b_{M,Z,i}} X_{q_i}^{b_{M,X,i}} \\
            & \propto Z_{q_M}^{b_{M,Z,M}}X_{q_M}^{b_{M,X,M}} \prod_{i=1}^{M-1} Z_{q_i}^{b_{M,Z,i}} X_{q_i}^{b_{M,X,i}} \\
            & = \prod_{i=1}^{M} Z_{q_i}^{b_{M,Z,i}} X_{q_i}^{b_{M,X,i}},
        \end{split}
    \end{equation}
    where the second equality is by Eq.~\eqref{eq:tilde-p}, the third equality uses the fact $C''$ does not act on qubits $q_1,\dots,q_{M-1}$, the third equality (up-to-sign) uses the fact circuit $C''$ localizes $\tilde P_M$ in qubit $q_M$.
    Finally, for $m=1,\dots, M-1$ we have 
    \begin{equation}
        \begin{split}
        C P_mC^\dagger &= C''C'P_m(C')^\dagger (C'')^\dagger \\
            &\propto C''   \prod_{i=1}^{M-1}  Z_{q_i}^{b_{m,Z,i}} X_{q_i}^{b_{m,X,i}} (C'')^\dagger \\
            & =C''  (C'')^\dagger \prod_{i=1}^{M-1}  Z_{q_i}^{b_{m,Z,i}} X_{q_i}^{b_{m,X,i}} \\
            &= \prod_{i=1}^{M-1}  Z_{q_i}^{b_{m,Z,i}} X_{q_i}^{b_{m,X,i}},
        \end{split}
    \end{equation}
    where the first equality (up-to-sign) comes from the definition of $C'$, the next equality comes from the fact $C''$ is not acting on qubits $q_1,\dots, q_{M-1}$, and the last equality coming from the fact $C''$ is unitary.
\end{proof}
Note that the proof is not tight in general, as $4^M-1$ Paulis of the form 
\begin{equation}
  \prod_{i=1}^M Z_{q_i}^{b_{m,Z,i}} X_{q_i}^{b_{m,X,i}},
\end{equation}
for arbitrary $b_{m,Z,i}, b_{m,X,i} \in \{0,1\}$ with at least one nozero value are already compressed into $M$ qubits. On the other hand, we found it to be tight in the worst case scenario as Paulis $Z_{q_1}, \dots, Z_{q_M}$ cannot be compressed to less than $M$ qubits.

\section{Details on numerical experiments}

In all the experiments, we compare our compilation method against our implementations of Steiner Synthesis~\cite{miller2024treespilation} and Lazy Synthesis~\cite{martiel2022architecture}. In all the methods we used Mehlhorn algorithm for finding a Steiner tree~\cite{mehlhorn1988faster}. When using PauliOpt either as a comparison or as a part of our compilation routine, we used the implementation available in~\cite{paulioptcode}. In all the cases in which we used our compilation method, we used precomputed Clifford databases as explained in Sec.~\ref{sec:clifford-database}.

\subsection{Multi-Pauli Lazy Synthesis -- random circuits} \label{sec:numerics-general}

For results presented in Figs~\ref{fig:per-paulis} and~\ref{fig:per-qubits} we considered Jordan-Wigner mapping and Bravyi-Kitaev on a path graph. The initial layout was chosen to be consecutive numbers 0, \dots, $N-1$ for $N$ qubits. For each number of qubits considered, the ansatze were generated randomly by selecting 300 Pauli strings $P_k$ being random 4-products of Majorana strings, to give a reminiscent of the ADAPT-VQE pool presented in~\cite{miller2024treespilation}, which later were added to the ansatz in the form of an exponentiation $\exp(i\theta_kP_k)$. Then, to obtain circuits with different numbers of Paulis, thus complexity, we chose subcircuits consisting of $M\leq 300$ exponentiations in the ansatz.

The SS compilation method was fine-tuned with Qiskit transpiler with optimization level 1. MPLS method was implemented with the following details:
\begin{enumerate}
    \item Pauli lists of length $M$ were selected so that for each list there were  $M \leq $ 3 Paulis that span all the other Paulis in the list
    \item Given a list of Paulis $P_1,\dots,P_M$, circuits for each sublists $P_1,\dots,P_k$ for each $k=1,\dots, M$ were found. The sublist of Paulis was selected so that $\#\text{CNOT}/k$ was minimized.
    \item The compression was done using a pre-computed database and the leaf to be removed was selected randomly. The information was compressed from 4 into 3 qubits.
    \item During the compression phase, whenever a Pauli appeared in CER, and commutativity relations for Paulis in the list allowed for implementing the Pauli, it was implemented.
    \item Once the subgraph containing the information about Paulis consists of 4 nodes only, the remaining Paulis from the list were implemented in batches of 3. 
\end{enumerate}
The Clifford circuits that appeared at the end of LS and MPLS methods were resynthesized using Pauliopt method. 

\subsection{Multi-Pauli Resetted  -- double excitations} \label{sec:numerics-de}
For JW mapping a path graph was considered with layout as explained in the previous subsection. Random PPTT mapping was chosen. The results were generated on IBM Brisbane connectivity. We started by taking a minimum height ternary tree with a root at node labeled 62, and creating a random PPTT mapping on that tree. The root was selected because it is in the middle of the heavy-hexagonal graph,  which in turn guarantees $O(\sqrt{n})$-local Majorana strings for large $n$. The initial layout was coming from PPTT mapping. For such a PPTT mapping, we took a random double excitation operator and synthesized it using three different methods.
\begin{enumerate}
    \item Steiner synthesis as explained in Sec.~\ref{sec:numerics-general}, except that it was first prepended with the trotterization of the 8-Pauli double excitations. The order of trotterization was carried to maximize the CNOT cancellation, by reducing the problem of finding a good order to the Hamiltonian Path problem, in which the distance between the nodes was the number of CNOTs that were not cancelled out.
    \item Lazy synthesis implemented exactly as in the previous subsection with arbitrary trotterization and final Clifford circuits for final resynthesis via PauliOpt.  
    
    \item MPR prepended with the same trotterization order as for lazy synthesis. In this case:
    \begin{enumerate}
        \item We always compressed all 8-Paulis together, in which the leaf to be removed was selected randomly. The compression components were removing qubit by qubit each time compressing the information from some 5 qubits to a subset of 4 qubits.
        \item Once the subgraph containing information about Pauli consisted of 4 nodes, the Paulis were implemented in batches of 3 Paulis. Batching was done because finding Clifford circuits that would implement all 8 Paulis turned out to be infeasible.
        \item After all Paulis were implemented, the Clifford gates coming from implementation Clifford circuits were resynthesized using PauliOpt. Finally, the inverse of the Clifford part from compression was implemented to finalize resetting. 
    \end{enumerate}
\end{enumerate}

\subsection{Multi-Pauli Lazy Clifford synthesis}  \label{sec:numerics-clifford}

For the results presented in Fig.~\ref{fig:clifford}, $K$-gates Clifford circuits were constructed heuristically by repeating the following procedure $K$ times. First, a pair of different qubits was randomly selected. Then on each of the qubits, a random 1-qubit Clifford gate is implemented. Afterward, a CNOT gate was implemented on these qubits. The procedure was repeated $K$ times, and this way we were able to steer the complexity of the Clifford circuits.

All the Clifford circuits were synthesized for a path graph. We compared the following methods:
\begin{enumerate}
    \item Clifford synthesis method from \cite{winderl2023architecture} which we refer to as PauliOpt. We took the original implementation from \cite{paulioptcode}.
    \item Variant of PauliOpt strategy, in which the logical qubits possibly correspond to different physical qubits i.e., the synthesis is done up to a sequence of SWAPs on the right-hand side. In each step a pair of logical $q$ and physical qubits $ q'$ are taken which requires the smallest number of gates to compute first $Z_q$, then $X_q$ on physical qubits $q'$. 
    \item Multi-Pauli Clifford Synthesis method where the order of logical qubits is preserved. The resynthesis procedure works as follows:
    \begin{enumerate}
        \item For each qubit that does not disconnect the subgraph, two circuits were generated: the first one works by implementing $Z_q$ on qubit $q$ and later implementing $X_q$ on qubit $q$, Pauli by Pauli as in the previous two variants. The second circuit was generated by first compressing information $Z_q,X_q$ into the qubit removing other qubits one by one in a multi-Pauli fashion. Note that after the compression is done, $Z$, $X$ and $Y$ registers containing $Z_q$, $X_q$, $Y_q$ up to permutation and sign. This can be fixed by applying 1-qubit operations. This method is similar in spirit to~\cite{dilkes2024greedy}, although differ in technical details   
        \item All the circuits are investigated and the qubit $q$ with the corresponding circuit having a smaller number of CNOTs is chosen. 
        \item qubit $q$ is removed from the subgraph and the procedure is repated until the subgraph has no nodes.
    \end{enumerate}
    \item Multi-Pauli Lazy Clifford synthesis method in which the logical qubits possibly correspond to different physical qubits. The same procedure as in the previous point is implemented, except that Paulis $Z_q$ and $X_q$ might be uncomputed to possibly different qubits $q'$, and thus all possible pairs of $(q,q')$ have to be investigated.
\end{enumerate}

\subsection{Multi-Pauli lazy synthesis -- treespilation} \label{sec:numerics-treespilation}

For the results presented in Table~\ref{tab:treespilation}, we used the exact implementation of the treespilation algorithm from \cite{miller2024treespilation}. For the simulated annealing within the treespilation algorithm, the following parameters were taken: cooling final values = 1, cooling initial value = 50, and cooling rate = 0.9995.  The treespilation algorithm was combined with SS, LS, and MLPS. In the case of SS, compilation of the ansatz was repeated for 8 times. For each of the methods, the treespilation algorithm itself was repeated 16 times and the best result was taken. For the molecules, ground state ansatze prepared by ADAPT-VQE algorithm using majoranic pool for BODIPY-4 from~\cite{bodipy}, and  H$_6$, LiH, N$_2$, BeH$_2$ molecules were used. For the connectivity, heavy-hexagonal was taken. We used SS and LS compilers as described in Sec.~\ref{sec:numerics-general}.

MPLS compiler was designed taking into account the following steps:
\begin{enumerate}
    \item If a Pauli to be implemented is already present in the intermediate representation, it is implemented. This step is repeated until no such Pauli from the group is found.
    \item If all the Paulis to be implemented from a group are localized within 4 qubits forming a connected induced subgraph, they are all implemented using a Clifford circuit from the implementation database in batches of size 3.
    \item If a Pauli to be implemented is localized within 2 qubits forming a connected subgraph, it is implemented. This step is repeated until no such Pauli from the group is found.
    \item If a Pauli to be implemented is localized within 3 qubits forming a connected subgraph, it is implemented. This step is repeated until no such Pauli from the group is found.
    \item If a Pauli to be implemented is localized within 4 qubits forming a connected subgraph, it is implemented. This step is repeated auntil no such Pauli from the group is found.
    \item The compression component is executed, in which a qubit from the subgraph from which the information is removed is chosen to be the one that requires the smallest number of 2-qubit gates. After this step, the procedure repeats starting from step~1. 
\end{enumerate}
The Pauli strings were selected so that there were at most 3 Paulis that were algebraically generating all the Paulis in a group. Given a list of Paulis $P_1,\dots,P_M$, a different sublist $P_1,\dots,P_m$ for $m\leq M$ was compressed and implemented Paulis were selected to be those for which $\#\text{CNOT}/m$ was minimized. Resetting is only performed at the end of the circuit and uses the method described in Sec.~\ref{sec:numerics-clifford}.

\end{document}